\documentclass[a4paper,10pt]{IEEEtran}

\pdfoutput=1
\IEEEoverridecommandlockouts
\makeatletter
\def\endthebibliography{%
	\def\@noitemerr{\@latex@warning{Empty `thebibliography' environment}}%
	\endlist
}
\makeatother

\IEEEoverridecommandlockouts
\usepackage{setspace}
\usepackage{multirow}
\usepackage{lipsum}
\usepackage{amsfonts}
\usepackage{array}
\usepackage{amsmath,cases,amssymb}
\usepackage{amsmath}
\usepackage{amsthm}
\usepackage{graphicx}
\usepackage{cite}
\usepackage{mathtools}
\usepackage{subfigure}
\usepackage{epsfig}
\usepackage{url}
\usepackage{algorithm}
\usepackage{algorithmic}
\usepackage{epstopdf}
\usepackage{color}
\usepackage{makecell}

\usepackage{lipsum}

\usepackage[justification=centering]{caption}

\DeclareGraphicsExtensions{.eps,.pdf}

\newcommand{\bY}{\mathbf{Y}}

\newcommand{\bB}{\boldsymbol{B}}

\newcommand{\bI}{\mathbf{I}}
\newcommand{\bh}{\mathbf{h}}

\newcommand{\bu}{\mathbf{u}}

\newcommand{\bF}{\mathbf{F}}

\newcommand{\ba}{\boldsymbol{a}}
\newcommand{\bs}{\boldsymbol{s}}

\newcommand{\by}{\mathbf{y}}

\newcommand{\bN}{\mathbf{N}}
\newcommand{\bn}{\mathbf{n}}

\newcommand{\E}{\mathbb{E}}

\newcommand{\Ncal}{\mathcal{N}}

\newcommand{\Tcal}{\mathcal{T}}
\newcommand{\T}{\mathtt{T}}
\newcommand{\D}{\mathtt{D}}
\newcommand{\R}{\mathtt{R}}
\newcommand{\N}{\mathtt{N}}
\newcommand{\s}{\mathtt{s}}
\newcommand{\com}{\mathtt{c}}
\newcommand{\p}{\mathtt{p}}
\newcommand{\tot}{\mathtt{tot}}
\newcommand{\bvarpi}{\boldsymbol{\varpi}}

\makeatletter
\newcommand{\doublewidetilde}[1]{{%
		\mathpalette\double@widetilde{#1}}}
\newcommand{\double@widetilde}[2]{%
	\sbox\z@{$\m@th#1\widetilde{#2}$}%
	\ht\z@=.5\ht\z@
	\widetilde{\box\z@}}
\makeatother

\newtheorem{lemma}{Lemma}

\begin{document}
\newcommand{\pp}[1]{\textcolor{red}{#1}}
\newcommand{\phuc}[1]{\textcolor{blue}{#1}}
\newcommand{\forest}[1]{\textcolor{green}{#1}}

\title{Digital Twin for Autonomous Guided Vehicles based on  
Integrated Sensing and Communications}

\author{Van-Phuc Bui, Pedro Maia de Sant Ana,  Soheil Gherekhloo, Shashi Raj Pandey, Petar Popovski
\thanks{V.-P Bui, S.R. Pandey, and P. Popovski (emails: \{vpb, srp, petarp\}@es.aau.dk) are all with the Department of Electronic Systems, Aalborg University, Denmark. P. M. de Sant Ana is with Corporate Research and Soheil Gherekhloo is with Cross-Domain Computing Solutions, both at Robert Bosch GmbH, Germany (email: \{Pedro.MaiadeSantAna, soheil.gherekhloo\}@de.bosch.com). This work was supported by the Villum Investigator Grant ``WATER'' from the Velux Foundation, Denmark.}
}

\maketitle
\thispagestyle{empty}
\begin{abstract}
    This paper presents a Digital Twin (DT) framework for the remote control of an Autonomous Guided Vehicle (AGV) within a Network Control System (NCS). The AGV is monitored and controlled using Integrated Sensing and Communications  (ISAC).  In order to meet the real-time requirements, the DT computes the control signals and dynamically allocates resources for sensing and communication. A Reinforcement Learning (RL) algorithm is derived to learn and provide suitable actions while adjusting for the uncertainty in the AGV’s position. We  present closed-form expressions for the achievable communication rate and the Cramer–Rao bound (CRB) to determine the required number of Orthogonal Frequency-Division Multiplexing (OFDM) subcarriers, meeting the needs of both sensing and communication. The proposed algorithm is validated through a millimeter-Wave (mmWave) simulation, demonstrating significant improvements in both control precision and communication efficiency.
\end{abstract}
% \begin{IEEEkeywords}
%     Digital twin, ISAC, Reinforcement Learning, Internet of Things, Dynamic Systems.
% \end{IEEEkeywords}

%%%%%%%%%%%%%%%%%%%%%%%%%%%%%%%%%%%%%%%%%%%%%%%%
\vspace*{-18pt}
\section{Introduction}\label{sec:intro}
\vspace*{-2pt}
%%%%%%%%%%%%%%%%%%%%%%%%%%%%%%%%%%%%%%%%%%%%%%%%

Smart manufacturing in Industry 4.0 requires collection of extensive real-time data from a wide range of wireless sensors \cite{tang2015tracking}. Unlike the traditional simulation tools or optimization techniques, digital twin (DT) models convert these large datasets into predictive models \cite{10574266}. This enables simulation of various control strategies, thereby facilitating real-time interactions and decision-making for system operators~\cite{9899718}.

A network control system (NCS) is used to power up a DT, where the latter can be situated either at the edge, tailored to local physical systems, or in the cloud, covering extensive physical systems. The network comprises sensor devices and/or central/distributed units integral to a 5G system or beyond \cite{dahlman20205g}. The acquired knowledge from the DT model serves twofold purposes: (1) controlling the physical world; and (2) provision of monitoring system states. Traditionally, tracking the physical world relies upon wired/wireless connections for internal communications among sensors, controllers, and actuators. However, this can limit the agility and scalability of system because different sensors provide varying sensing accuracy and associated costs. Additionally, in a complex environment comprising multiple sensors and remote devices operating simultaneously, there is a significant burden on processing various types of signals and managing interference. To overcome these challenges, we advocate the use of Integrated Sensing and Communications (ISAC) \cite{mandelli2023survey}, which leverages compact access points (APs) to enable simultaneous communication and monitoring physical world. 

There has been substantial research on ISAC covering foundational concepts, algorithms, analysis, and demonstrators~\cite{mandelli2023survey, liu2022survey}\cite{hexaxii_project}. It has been shown in~\cite{nguyen2023multiuser} that the communication rate of ISAC systems can be maximized within radar performance constraints. 
%Researchers has also focused on the use of Reinforcement Learning (RL) in DT systems to perform various tasks~\cite{Bui_2024}. 
DT frameworks are also applied in monitoring Automated Guided Vehicles (AGVs) by effectively selecting and scheduling UWB sensors for efficient demonstration in industry \cite{bui2024digital}. However, the aforementioned works \cite{mandelli2023survey, liu2022survey,hexaxii_project, nguyen2023multiuser, bui2024digital} and their references do not consider the reciprocal influence between ISAC and control performance in the physical world based on communication requirements.

Here we examine the DT architecture for tracking dynamic changes in system parameters and efficiently controlling the AGV. We formulate an optimization problem to balance subcarrier allocation between sensing and communication, aiming to maintain the confidence of the DT's system estimates while simultaneously optimizing the AGV control policy. We devise an RL agent that learns to perform actions while managing the uncertainty estimation of the state. A closed-form solution for subcarrier allocation is presented, tailored to meet the required AGV position uncertainty from both the RL agent and the DT. In terms of communication, we calculate the amount of subcarriers required to meet the throughput demand. We introduce an efficient solution to allocate the available subcarriers within a mmWave environment, balancing the needs of sensing and communication. The numerical results to validate our theoretical findings, offering a superior performance of the proposed subcarrier allocation scheme.

\vspace{-10pt}
%%%%%%%%%%%%%%%%%%%%%%%%%%%%%%%%%%%%%%%%%%%%%%%%
\section{Digital Twin Architecture}
%%%%%%%%%%%%%%%%%%%%%%%%%%%%%%%%%%%%%%%%%%%%%%%%
\subsection{System model}
We explore the application of ISAC to support NCSs within a DT framework for a AGV, as illustrated in Fig.~\ref{fig:ISAC_system}. The DT model's objective is to uphold a precise estimate of AGV's state  and offer optimized sequence of actions to be executed in the physical realm based on its beliefs about states. The AGV is monitored và controlled though set of APs deployed at the physical world. These APs synchronize their DTs,
including locations and power budget status, with the Cloud platform with a high reliability and periodic synchronization overhead. Based on the current model state and the update from the AGV, the DT model updates the AGV's active state, control policy, control signal, and schedules OFDM subcarriers for sensing and communication at the APs. The control problem considered here is known as the Mountain-Car Continuous problem \cite{moore1990efficient}, inspired by control systems for autonomously operating AGVs. The link between the AP and the cloud is assumed to be perfect. At the start of each query interval (QI), the scheduled AP uses ISAC within $\Delta_\mathtt{ISAC}$ interval to sense the state and communicate with the AGV. The DT continuously monitors the AGV’s state through APs and dynamically adjusts the control inputs to optimize its path. Our approach does not use no specific techniques for estimating the future state of the AGV, such that the control signals are calculated based on the previous state of the AGV.
This is reasonable when the indoor AGV operates at a low speed and the QI length is short.
The controller generates a command and sends it over a downlink channel. The application output for actuator control retrieves the latest commands from memory and applies them to the AGV. The AGV's state at any QI $t$ is represented by $\bs_t = [x_t, v_t]$, where $x_t/v_t$ denotes the position/velocity, respectively, at QI $t$. The control signal $a_t$ indicates the applied force. The system dynamics is modeled as 
\begin{IEEEeqnarray}{lll}\label{dynamic_model}
    \bs_{t+1} = f(\bs_t) + \bB a_t + \bu_t, \forall t \in \Tcal,
\end{IEEEeqnarray}
where $f(\bs_{t})= \begin{bmatrix}
	x_t + v_{t}	\\ 
	v_t -\varphi \cos(3x_{t})	
\end{bmatrix},
\bB = [0, \vartheta]^\top,$
with the constants $\varphi = 0.0025$ and $\vartheta = 0.0015$ \cite{GymnasiumMountainCarContinuous}.  $\mathbf{u}_t\sim \mathcal{N}(\mathbf{0},\mathbf{C}_{\mathbf{u}})$ indicates the process noise.

\vspace{-10pt}
\subsection{ISAC model}
We consider a mono-static configuration featuring quasi co-located  transmitters (TX) and receivers (RX) with equal characteristics, comprising $R$ rows and $C$ columns, and separated by distances $\Delta r$ and $\Delta c$, respectively \cite{wild2023integrated}. The budget at AP  includes total $N$ subcarriers and $M$ OFDM symbols. The DT model attempts to schedule $n_{\s,t}$ subcarriers for sensing and $n_{\com,t}\leq N-n_{\s,t}$ subcarriers for communication.
\begin{figure}
    \centering
    \includegraphics[width = 0.5\textwidth]{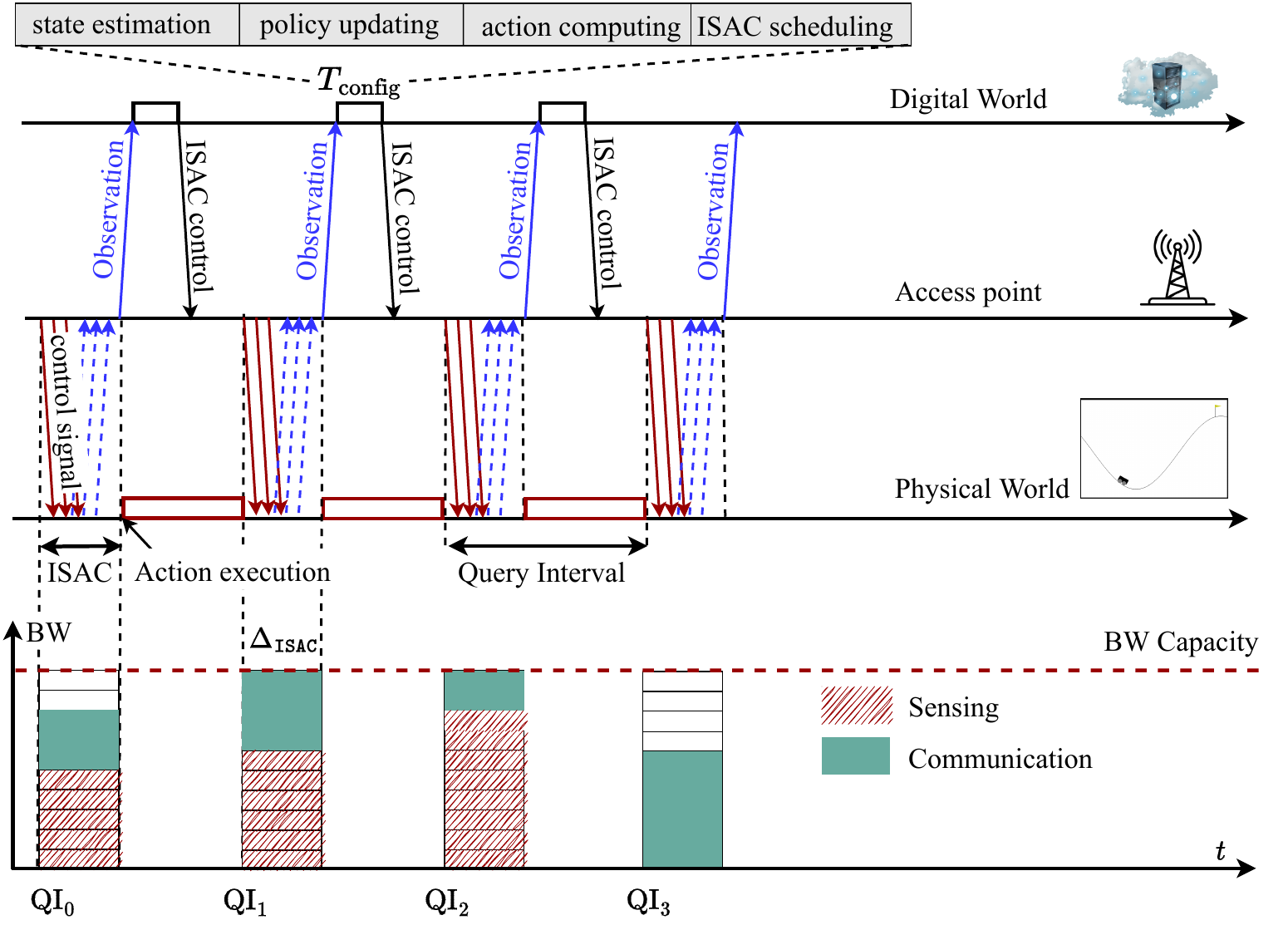}
    \vspace{-15pt}
    \caption{The ISAC-based Digital Twin diagram.}
    \label{fig:ISAC_system}
    \vspace{-15pt}
\end{figure}
% \vspace{-10pt}
\subsubsection{Sensing model}
\begin{figure}
\begin{minipage}{0.24\textwidth}
    %		\centering
    \includegraphics[trim=0.1cm 1cm 1.cm 1cm, clip=true, width=1.75 in]{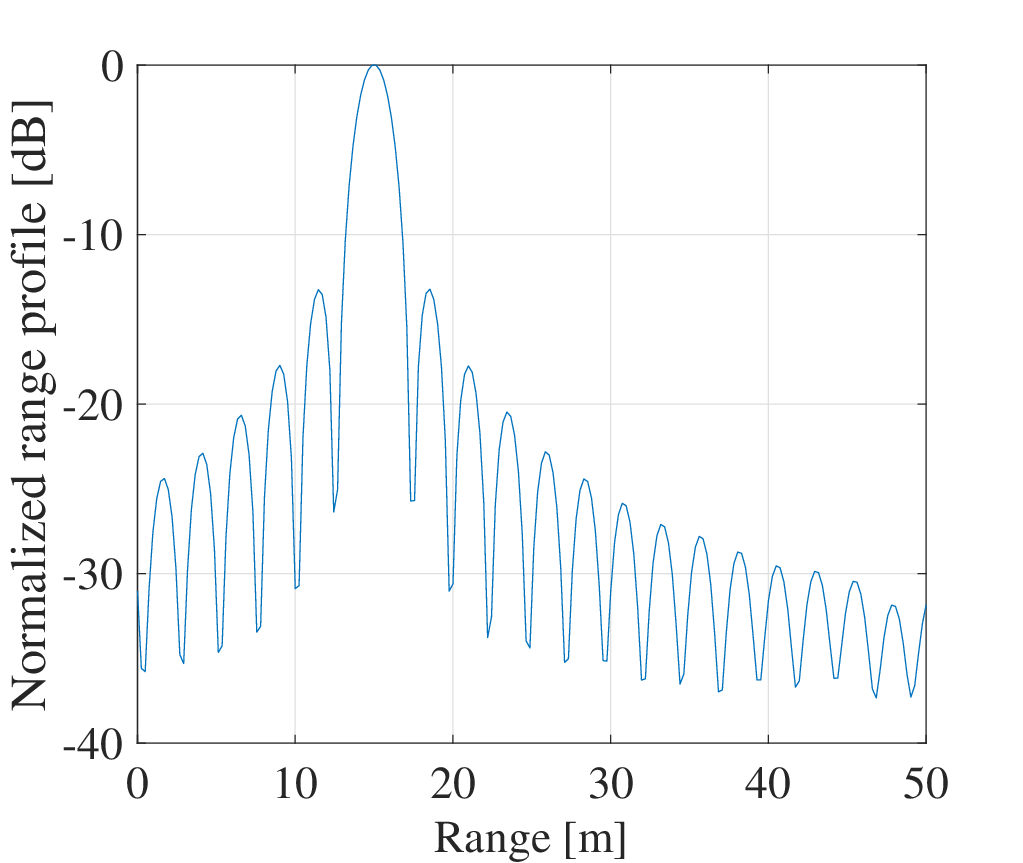} \\ 
    \vspace*{-5pt}
    \centering {\footnotesize$(a)$}
    \vspace*{-3pt}
\end{minipage}
\begin{minipage}{0.24\textwidth}
    %		\centering
    \includegraphics[trim=0cm 0cm 0.cm 0cm, clip=true, width=1.76 in]{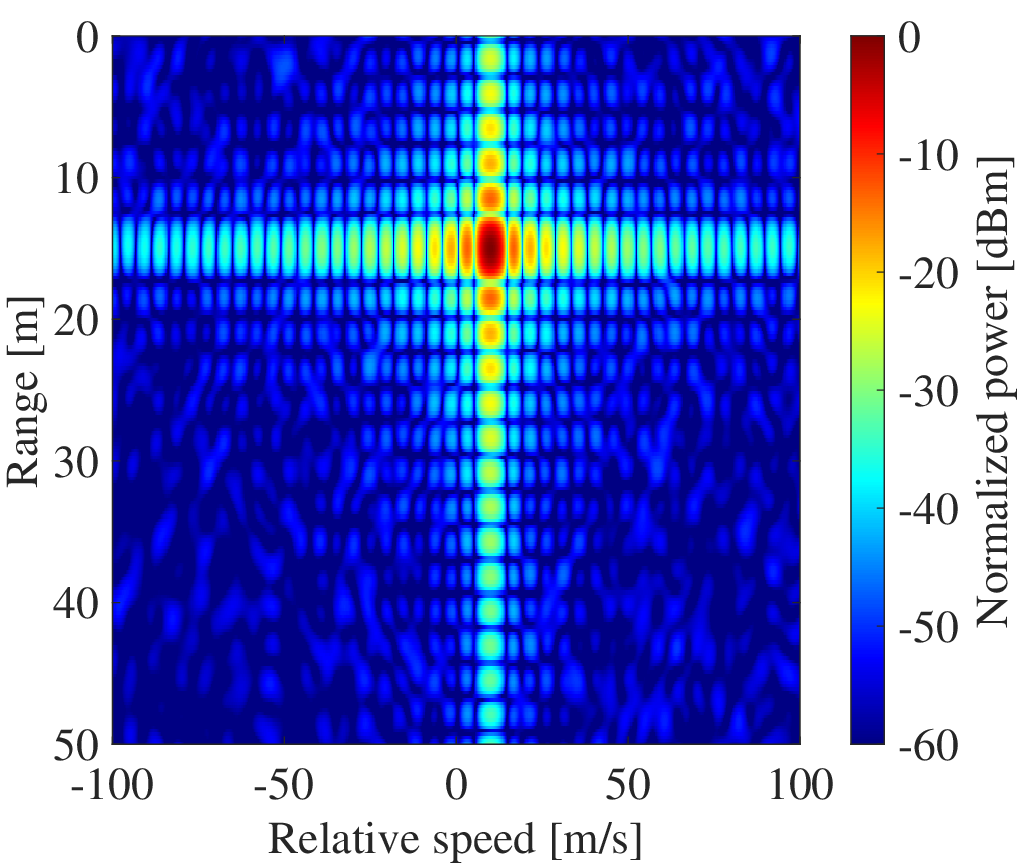} \\ 
    \vspace*{-5pt}
    \centering {\footnotesize$(b)$}
    \vspace*{-3pt}
\end{minipage}
\captionsetup{format=plain, justification=justified, width=1\linewidth}
\caption{\small{Example with one object (mmWave setup, $N = 512, M = 128$): $(a)$ Normalized range profile; and $(b)$ a periodogram.}}
\label{fig_sensing}
\vspace{-20pt}
\end{figure}
The received reflected signal $r_t$  has the form $r_t = b_t \varrho_{t - \tau_t} e^{j 2 \pi f^\D_t,t} e^{j \tilde{\psi}_t} + \tilde{z}_t$, where the transmitted signal is $\varrho_t$, the delay $\tau_t = 2\frac{r_t}{c}$, the Doppler shift $f^\D_t = 2\frac{v_t}{c}f_c$, an addition random phase rotation $\tilde{\psi}_t$, and AWGN $\tilde{z}_t \sim\Ncal(0, \sigma_{\tilde{z}}^2)$. The attenuation, using a point scatter model, is $b_t = \sqrt{\frac{c^2 \Psi}{(4\pi)^3 r_t^4 f_c^2}}$. We denote by $\Psi, r, \lambda, f_c, c$ the radar cross section, range, wavelength,  central frequency, and speed of light, respectively. In LoS sensing, the received power $P^\mathtt{R}_t$ is:
\begin{align}
    P^\mathtt{Rs}_t = \frac{P^\mathtt{T} G^\T}{4\pi r^2_t}\Psi \frac{1}{4 \pi r^2_t}\frac{G^\R \lambda^2}{4\pi} = {P^\T G^\T G^\R} b_t^2, \forall t\in\Tcal,
\label{eq:RxPower}
\end{align}
where  $P^t, G^t, G^r$ are transmitted power, transmit antenna array gain, receive antenna array gain, respectively, by assuming a perfect focus on the target with the transmit beam. The received noise power is $ P^\N_t = (N_0 F) (n_{\s,t} \Delta f)$, 
% \begin{equation}
%     P^\N_t = (N_0 F) (n_{\s,t} \Delta f),
% \end{equation}
where $N_0$ is the noise spectral density, $F$ is noise figure, and $\Delta f$ the subcarrier spacing. The sensing Signal-to-Noise (SNR) ratio is  calculated as \cite{braun2014ofdm}:
\begin{align}
    \gamma_{\s,t} = \frac{P^{\R\s}}{P^\N} n_{\s,t}M = {P^\T G^\T G^\R}  \Psi \frac{c^2M}{(4\pi)^3 r^4_t f_c^2 N_0F\Delta f}.
    \label{eq_SNR}
\end{align}
% \subsubsection{Estimation in Sensing}
Assuming the transmitted OFDM frame is presented by $\bF^\T \in \mathcal{F}^{N\times M}$ where $\mathcal{F}$ is the modulation alphabet. Specially, each row of $\bF^\T$ represents a sub-carrier; each column represents an OFDM symbol of the transmitted frame. Then the element $(n,m)$ of the received frame matrix $\bF^\R_t$  at intervals of length $T_0$ is computed by $[\bF^\R_t]_{(n,m)} = b e^{j 2 \pi m T_0 f^D_t} e^{-j 2 \pi n \tau_t \Delta f} e^{j \tilde{\psi}_t} + [\mathbf{Z}_t]_{(n,m)}$. 
% \begin{equation}
% [\bF^\R_t]_{(n,m)} = b. e^{j 2 \pi m T_O f^D_t} e^{-j 2 \pi n \tau_t \Delta f} e^{j \tilde{\psi}_t} + [\mathbf{Z}_t]_{(n,m)}.
% \end{equation}
Using solution from \cite{5776640}, the periodogram of $\bF_t$ is formulated as $[\mathrm{Per}^{\mathbf{F}}_t]_{(n,m)} = \tfrac{ \left| \sum_{k=0}^{N_{\text{Per}}-1} \left( \sum_{l=0}^{M_{\text{Per}}-1} [\mathbf{F}_t]_{(k,l)} e^{-j 2 \pi \frac{l m}{M_{\text{Per}}}} \right) e^{j 2 \pi \frac{k n}{N_{\text{Per}}}} \right|^2}{NM}.$
% \begin{IEEEeqnarray}{lll}
%     [\mathrm{Per}^{\mathbf{F}}_t]_{(n,m)} = \tfrac{ \left| \sum_{k=0}^{N_{\text{Per}}-1} \left( \sum_{l=0}^{M_{\text{Per}}-1} [\mathbf{F}_t]_{(k,l)} e^{-j 2 \pi \frac{l m}{M_{\text{Per}}}} \right) e^{j 2 \pi \frac{k n}{N_{\text{Per}}}} \right|^2}{NM}. \nonumber
% \end{IEEEeqnarray}
AGV's estimated range and velocity are computed by $\bar{r}_t = ({\mathrm{Pe}^{\hat{n}}_t c})/({2 \Delta f N_{\text{Per}}}) $ and $\bar{v} = ({\mathrm{Pe}^{\hat{m}}_t c})/({2 f_c T_0 M_{\text{Per}}})$.
% \begin{equation}
% \bar{r}_t = \frac{\mathrm{Pe}^{\hat{n}}_t c}{2 \Delta f N_{\text{Per}}} \quad \text{and} \quad \bar{v} = \frac{\mathrm{Pe}^{\hat{m}}_t c}{2 f_c T_0 M_{\text{Per}}}.
% \end{equation}
Herein, $\mathrm{Pe}^{\hat{n}}_t$ and $\mathrm{Pe}^{\hat{m}}_t$ are respectively peak values of  the matrix $\mathrm{Per}^{\mathbf{F}}_t$. $N_\mathrm{Per}$ and $M_\mathrm{Per}$ are the number of IFFT and FFT in transformation, respectively. Fig.~\ref{fig_sensing} plots the normalized range profile and an example of a periodogram to sense the range and speed under mmWave setup.
% In addition, we have the physical azimuth
% and elevation from NAF dimensions as
% \begin{IEEEeqnarray}{lll}
% \bar{\phi}_t = \sin^{-1} \left( \frac{\lambda}{\Delta r} \zeta_t \right), \ 
%     \bar{\theta}_t = \sin^{-1} \left( \frac{\lambda}{\Delta c} \cos(\bar{\phi}_t) \ell_t \right), 
% \end{IEEEeqnarray}
% where $\bar{\phi}_t$ is the mean of sensed azimuth angle at time $t$. 
The Cramér-Rao Bound (CRB) establishes a minimum limit on the variance of unbiased estimators. Within the considered framework, we outline the uncertainty for position and velocity, specifically for the purposes of monitoring and estimating within the DT model. To achieve this, we first derive the CRB for estimating the target's range, velocity, and elevation angle as \cite{mandelli2023survey}
\begin{align}
\sigma_{r,t} &= \frac{c}{4\pi \Delta f} \sqrt{\frac{6}{(n_{\s,t}^2-1)\gamma_t^\s}} , \forall t\in \Tcal,
\label{eq:AccuracyRange} \\
\sigma_{v,t} &= \frac{c}{4\pi f_c T_0} \sqrt{\frac{6}{(M^2-1)\gamma_t^\s}} , \forall t\in \Tcal,
\label{eq:AccuracySpeed} \\
\sigma_{\theta,t} &=  \sin^{-1}\left(\frac{\lambda}{\Delta c}\cos({\phi}_t)(\ell_t\pm \sigma_{\mathtt{fx},t})\right) - {\theta}_t,  \forall t\in \Tcal\label{eq:AccuracyElevation},
% \sigma_x &=  \frac{1}{2\pi } \sqrt{\frac{6}{(C^2-1)\gamma(r)}}  \label{eq:AccuracyAzimuth}.
\end{align}
% where $\sigma_{\mathtt{fx},t} = \frac{1}{2\pi } \sqrt{\frac{6}{(R^2-1)\gamma_t}}$ 
where $\sigma_{\mathtt{fx},t}^2 ={6}/({(R^2-1)4\pi^2\gamma_t})$
is the normalized angular frequencys (NAFs) in horizontal direction of the receive array \cite{mandelli2022sampling}. The NAF $\ell_t$ is calculated by $\ell_t = ({\Delta c\sin({\theta}_t)})/({\lambda\cos({\phi}_t)})$.
% The horizontal NAF $\ell_t$ and vertical NAF $\zeta_t$ are respectively computed by $\ell_t = ({\Delta c\sin({\theta}_t)})/({\lambda\cos({\phi}_t)})$, and $\zeta_t = \Delta r/\lambda \sin({\phi}_t)$. 
Note that $\sigma_{v,t}$ does not depend on the number of allocated subcarriers and we focus on the effect of AGV's position uncertainty to calculate the control signal.
% \vspace{-5pt}
\subsubsection{Communication Model}
% \phuc{cho nay can xem lai ky luong cung voi cai ben tren de dam bao la ham cua khoang cach r, vi luc nay r chua biet}
% The SNR of the receiver at QI $t$ is derived as
% \begin{IEEEeqnarray}{lll} 
%     \gamma_{\com, t} = \frac{MP^\T G^\T c^2}{(4\pi r_t f_c)^2N_0F\Delta f }.
%     % R_t = n_{\com, t}\Delta f \log_2\left(1+ \frac{MP^\T G^\T c^2}{(4\pi r_tf_c)^2N_0F\Delta f }\right).
% \end{IEEEeqnarray}

With multiple TX antennas, the received signal at AGVfrom transmitted waveform $\boldsymbol{\omega}_t \in\mathbb{C}^{RC}$ is $ y_t = \bh_t^H\boldsymbol{\omega}_t + w_t$, 
% \begin{IEEEeqnarray}{lll}
%     y_t = \bh_t^H\boldsymbol{\omega}_t + w_t, 
% \end{IEEEeqnarray}
where $w_t\sim\Ncal(0, \sigma_{w,t}^2)$ is the AWGN with $\sigma_{w,t}^2 = (N_0F)(n_{\com,t}\Delta f)$. The frequency flat fading channel between AP and the user is $\bh_t=\beta^{1/2}_t\bar{\bh}_t \in\mathbb{C}^{RC}$, where $\beta_t = \frac{n_{\com,t} G^\T c^2}{(4\pi r_t f_c)^2}$ is the large-scale fading and $\bar{\bh}_t$ the small-scale fading. 
We assume it is independent between communication and sensing process, leading to the channel is first estimated via uplink training, which is then used for the downlink transmission. Let $\tau^\com_\p$ and $\tau^\com_\tot$ be the length of  the pilot sequences and coherence interval in samples, respectively, with $\tau^\com_\p<\tau^\com_\tot $.  When AGV send their pilot sequence $\sqrt{\tau^\com_\p}\bvarpi\in\mathbb{C}^{\tau^\com_\p}$, the received pilot signal at the AP is given by $ \bY_t =  \sqrt{\tau^\com_\p p_\p}\bh_t\bvarpi^H + \bN_{\p,t}$, 
% \begin{IEEEeqnarray}{lll}
%     \bY_t =  \sqrt{\tau^\com_\p p_\p}\bh_t\bvarpi^H + \bN_{\p,t},
% \end{IEEEeqnarray}
where $\bN_{\p,t}\in\mathbb{C}^{RC \times \tau_\p^\com}$ is the AP noise, whose entries are independent and identically distributed (i.i.d.) $\Ncal(0,\sigma^2)$, and $p_\p$ is the average power of the training symbols. To estimate $\bh_t$, we project $\bY_t$ onto $\bvarpi$ as $\by_t = \bY_t\bvarpi = \sqrt{\tau^\com_\p p_\p}\bh_t + \bn_{\p,t},$ 
% \begin{IEEEeqnarray}{lll}
%     \by_t = \bY_t\bvarpi = \sqrt{\tau^\com_\p p_\p}\bh_t + \bn_{\p,t},
% \end{IEEEeqnarray}
where $\bn_{\p,t} = \bN_{\p,t}\bvarpi$ is additive noise at the AP section, which is Gaussian $\bn_{\p,t} \sim\Ncal(\mathbf{0}, \sigma^2_\p\mathbf{I}_{RC})$. At a result, the MMSE channel estimate is $\hat{\bh}_t = ({ \sqrt{\tau^\com_\p p_\p}\beta_t\bY_t\boldsymbol{\varpi}})/({ {\tau^\com_\p p_\p}\beta_t+\sigma_\p^2})$. 
% \begin{IEEEeqnarray}{lll}
%     \hat{\bh}_t = \frac{ \sqrt{\tau^\com_\p p_\p}\beta_t\bY_t\varpi}{ {\tau^\com_\p p_\p}\beta_t+\sigma_\p^2}.
% \end{IEEEeqnarray}
In another word, $\hat{\bh}_t$ follows $\hat{\bh}_t \sim\Ncal(\mathbf{0}, \sigma^2_{\hat{h},t}\bI_{RC})$, where $\sigma^2_{\hat{h},t}=\tau^\com_\p p_\p\beta_t^2/(\tau^\com_\p p_\p\beta_t + \sigma_\p^2)$. It is noted that the channel estimation error $\mathbf{e}_t = \bh_t-\hat{\bh}_t$ follows $\mathbf{e}_t \sim\Ncal(\mathbf{0}, \epsilon^2_t\bI_{RC})$ whose $\epsilon_t^2 = \beta_t-\sigma^2_{\hat{h},t}$. 
The achievable SNR at the AGV is $\gamma_{\com, t} = ({P^\T\sigma_{\hat{h},t}^2})/({P^\T \epsilon_t^2+\sigma_{w,t}^2})$. 
% \begin{IEEEeqnarray}{lll}
%     \gamma_{\com, t} = \frac{P^\T\sigma_{\hat{h},t}^2}{P^\T \epsilon_t^2+\sigma_{w,t}^2}.
% \end{IEEEeqnarray}
We approximate the achievable communication rate by $R_t(n_{\com,t}) = \bar{\tau}_t n_{\com,t}\Delta f\log_2(1+\gamma_{\com,t})$ with $\bar{\tau}_t = ({\tau_{\tot, t} - \tau_{\p,t}})/{\tau_{\tot,t}}$.

\vspace{-5pt}
\subsection{Problem Formulation}
% The DT model's objective is to uphold a precise estimate of $\PA$'s state  and offer the optimal sequence of actions to be executed in the physical realm based on its beliefs about states. 
% A fundamental distinction within our system lies in its thorough consideration and evaluation of real-world environments, which often involve state observations that are characterized by noise or corresponded costs. 
% Herein, the predicted estimator $\hat{\bs}_t$ of ${\bs}_t$ is modeled with 
% %	\begin{align}
%     %		p(\bs_t)\sim \mathcal{N}(\hat{\bs}_t, \bPsi_t), n \in\Ncal.
%     %	\end{align}
% $p(\bs_t)\sim \mathcal{N}(\hat{\bs}_t, \bTheta_t), n \in\Ncal$.  
% The MSE of the estimator is $ \text{MSE}_{} = \mathbb{E}\big[||\bs_t-\hat{\bs}_t||^2_2\big], t\in\Tcal$. 
% % \begin{align}
% %     \text{MSE}_{} = \mathbb{E}\big[||\bs_t-\hat{\bs}_t||^2_2\big], t\in\Tcal. 
% % \end{align}
% Our goal is to optimize AGV control signals and resource allocation for sensing and communication. 
Our goal is to maintain optimal control signals by implementing a control policy that meets the AGV's maximum allowable position variance and optimizes resource allocation for sensing and communication.
The maximal acceptable variance for AGV's position is defined as
% The maximal acceptable variance for AGV's position is
\begin{equation}\label{qos_condition}
    \sigma_{x,t}^2\leq {\xi}^2, \forall t\in \Tcal.
    % \sqrt{[\bTheta_t]_k} \leq {\xi}_k, \forall k \in\Kcal,
\end{equation}
% where $[\bTheta_t]_{k}$ is the $k$-th element of the diagonal of $\bTheta_t$.
$V^\pi(\bs_0)$ is defined as the value function of controlling AGV in \eqref{dynamic_model} under the control policy $\pi$, which is detailed in the next section. Our primary objective is to jointly minimize the total power consumption while ensuring the delivery of optimal control signals. Here $h\big(\{n_{\s,t}\}, \{n_{\com,t}\}, \{a_t\}\big)  = \Big[V^\pi(\hat{\bs}_0), -\sum_{t=0}^{\infty}(n_{\s,t} + n_{\com, t})\Big]^T,$ 
% \begin{IEEEeqnarray}{ll}\label{glob_objective}
% 	h\big(\{n_{\s,t}\}, \{n_{\com,t}\}, \{a_t\}\big)  = \Bigg[V^\pi(\hat{\bs}_0), -\sum_{t=0}^{\infty}(n_{\s,t} + n_{\com, t})\Bigg]^T,
% \end{IEEEeqnarray}
is a multi-objective function, encompassing the two performance metrics. We formulate the  optimization problem  as a joint design of control and scheduling $\{n_{\s,t}, n_{\com,t}\}$  while maintaining the confidence of DT's system estimate as follows
\begin{subequations} \label{glob_problem}
    \begin{alignat}{2}
	\underset{\{n_{\s, t}\},\{n_{\com, t}\}, \{a_t\}}  { \mathrm{maximize}} \ & h\big(\{n_{\s,t}\}, \{n_{\com,t}\}, \{a_t\}\big) \label{glob_problema}\\
	\mathrm{s.t.} \quad  & \mathbf{s}_t = f(\mathbf{s}_{t-1}) +  \bB a_{t-1} + \mathbf{u}_t, \label{glob_problemb}\\
   % & \mathbb{P}[\tau _{t,m} >  \tau^\mathrm{max}] \leq \varepsilon, \forall m\in\Mcal, t\in\Tcal,\label{glob_problemc}\\
   & R_t(n_{\com,t}) \geq \bar{R}, \forall t\in\Tcal,\label{glob_problemc}\\
   & \sigma_{x,t}^2 \leq {\xi}^2, \forall t\in\Tcal, \label{glob_problemd}\\
   & n_{\s, t} + n_{\com,t} \leq N, \forall t \in \Tcal.\label{glob_probleme}
    \end{alignat}
\end{subequations}
In order to maintain the communication Quality-of-Service (QoS) at every QI $t$, the throughput threshold $\bar{R}$ in 
\eqref{glob_problemb} is employed. The constraint \eqref{glob_probleme} ensures that the total number of subcarriers used $n_{\s, t} + n_{\com,t}$ does not exceed the system capacity $N$.
It is noted that we examine a context in which the belief vector is controlled via ISAC supported by DT before being employed by RL to propose the optimal action as an outcome, thereby allowing the agent to make  precise decisions. 

In the following, we introduce proposed a heuristic solution, which incorporates a two-step approach to effectively address  \eqref{glob_problem}: $(i)$ an uncertainty control RL algorithm is utilized to devise control actions for the physical world, effectively managing state estimation errors; $(ii)$ Subcarrier optimization algorithm is applied to identify suitable $\{n_{\s,t}, n_{\com,t}\}$, guided by the requirements of the RL model and DT. 

\vspace{-10pt}
%%%%%%%%%%%%%%%%%%%%%%%%%%%%%%%%%%%%%%%%%%%%%%%%%%%%%%
\section{Uncertainty Control POMDP}
The control problem in \eqref{glob_problem} is treated as a Partially Observable Markov Decision Process (POMDP), an extension of the Markov Decision Process (MDP) that includes sets of observations and observation probabilities due to partial and potentially inaccurate information. 
% In particular, a POMDP is represented by the 7-tuple $ \langle \Scal, \Acal, \Ocal, \mathtt{P}, \mathtt{O}, r, \gamma \rangle$, where $ \Scal $ is a finite set of states, $ \Acal $ is a set of actions, and $ \Ocal $ is a set of observations. At QI $ t $, the agent performs action $\ba_t$, transitioning from state $ \bs_t $ to $ \bs_{t+1} $ with transition probability $ \Ptt = \mathbb{P}[\bs_{t+1}|\bs_{t},\ba_t]$. The observation $ \bo_{t+1} $ occurs with probability $ \Ott = \mathbb{P}[\bo_{t+1}|\bs_{t+1},\ba_t] $, and the agent receives a reward $ r(\bs_{t}, \ba_t, \bs_{t+1}) $ such that $ r(\bs_{t}, \ba_t, \bo_{t+1}) \leq r^\mathrm{max} $.
In particular, the agent maintains an estimate vector $\hat{\bs}_t$, which describes the probability of being in a state $\bs_t$.  The policy  $\pi(\boldsymbol{\hat{s}},\boldsymbol{a})$ of the agent specifies an action $\ba_t$ based on this estimate, represented as $\pi(\hat{\bs}, \ba)$. Given an initial belief $\hat{\bs}_0$, the expected future discounted reward for the policy $\pi(\hat{\bs}, \ba)$ is expressed as $V^\pi(\hat{\bs}_0) = \mathbb{E}\left[\sum_{t=0}^{\infty}\gamma_t r(\bs_{t}, \ba_t, \bs_{t+1}) \mid \hat{\bs}_0, \pi\right]$,
where $0 < \gamma_t < 1$ is the discount factor. At QI $t$, the estimated state vector is $\hat{\bs}_t = [\hat{x}_t, \hat{v}_t]^\top \in \mathbb{R}^2$, governed by the conditional probability distribution function $\hat{\bs}_t \sim p(\hat{\bs}_t \mid \bs_t; {\eta}_t)$, where $\eta_t\in \mathbb{R}^1$ is parameterized by the accuracy of AGV's position. We define $\eta_{t} = {1}/{\sigma_{x,t}^2}$. As ${\eta}_{t}$ increases, the confidence in $\hat{x}_{t}$ improves, enabling the RL agent to make more accurate decisions. However, achieving high reliability for $\hat{x}_{t}$ requires low measurement error, which increases corresponding $n_{\s,t}$ and processing costs. 
% Given $\bs_{t}$ and $\boldsymbol{\eta}_t$, $\hat{s}_{t,k}$ are assumed to be statistically independent, allowing us to factorize the probability distribution as
% $p(\hat{\bs}_t \mid \bs_t; \boldsymbol{\eta}_t) = \prod_{k \in \Kcal} p(\hat{s}_{t,k} \mid \bs_t; \boldsymbol{\eta}_t)$ in terms of their factorization.  
For training $\pi(\boldsymbol{\hat{s}},\boldsymbol{a})$, we employ  Proximal Policy Optimization (PPO) with the actor-critic structure at RL agents that involves dividing the model into two distinct components, thus harnessing the strengths of both value-based and policy-based methods \cite{NIPS1999_6449f44a}. 
% To address the joint design problem with the objective of optimizing actions while minimizing communication energy, we implement the Deep RL (DRL) approach within the DT cloud environment, where  the  action and reward are are defined as follows.
To tackle the joint design problem aimed at optimizing actions while minimizing communication energy, we employ the Deep RL approach within the DT cloud environment, where the action and reward are defined as follows.

\subsubsection{Action Space Reformulation}
The action vector in the RL agent, structured as $\ba_t  = [{a}_{t},\eta_{t}]$, comprises control signals $a_t$ that affect the environment and $\eta_{t}$ for accuracy in estimated position. This study addresses optimizing the RL agent's state estimation accuracy, allowing continuous selection of $\eta_{t}$. The framework aims to reveal the informational value of observations for the task. 

%	\vspace*{-0.1cm}
\subsubsection{Reward Function Reformulation}
It is essential for the RL agent to not only pursue the primary goal defined by the problem but also to develop the capability to control the acceptable accuracy level $\eta_t$. The goal-oriented reward $r_t$ is adjusted into an uncertainty-based reward $\tilde{r}_t = r_t + \kappa \eta_t$, where $\kappa \geq 0$ serves as a weighting factor. The agent's dual objective is to maximize the original reward while concurrently minimizing the cost associated with observations.

% $\tilde{r}_t = f(r_t, {\eta}_t),$ 
% wherein $f(\cdot)$ is a monotonically non-decreasing function of $r_t$ and ${\eta}_t$. 
% In scenarios where a direct cost function, denoted as $c_k(\cdot)$, exhibits an upward trend with the accuracy of the observation $o_{t,k}$, a suitable additive formulation can be employed. Specifically, the modified reward, $\tilde{r}_t$, can be expressed by $\tilde{r}_t = r_t + \kappa  c(\eta_{t})$
% % \begin{align}\label{reward_shaping}
% %     \tilde{r}_t = r_t + \kappa  \sum_{k=1}^K c_k(\eta_{t,k}).
% % \end{align}
% Here, $c(\eta_{t})$ represents a non-increasing function of $\eta_{t}$, and $\kappa \geq 0$ is a weighting parameter. The primary objective of the agent is two-fold: maximizing the original reward while simultaneously minimizing the cost associated with  observations.
\begin{algorithm}[t]
    \begin{algorithmic}[1]{\fontsize{8.5pt}{9pt}\selectfont
            \protect\caption{ Proposed solution to the problem \eqref{glob_problem}} % \eqref{probGlobal}}
        \label{scheduling_alg}
        \global\long\def\algorithmicrequire{\textbf{Input:}}
        \REQUIRE System capacity $N$, DT requirement $({\xi_t, \bar{R}})$
        \global\long\def\algorithmicrequire{\textbf{Output:}}
        \REQUIRE $(n_{\com,t}, n_{\s,t})$, $\R_{\com,t}$, and $\hat{\bs}_t$\\ %Allocated subcarriers $(n_{\com,t}, n_{\s,t})$, Achieve data rate $\R_{\com,t}$, Estimated state $\hat{\bs}_t$\\
        % \vspace{3pt}
         \textit{Control strategy:} 
        \STATE Choose optimal action $\ba_t$ in estimated state $\hat{\bs}_t$ according to trained policy $\pi^*$ \\
        % \vspace{3pt}
         \textit{Subcarrier allocation for ISAC strategy:}
        \STATE Compute required $n_{\s,t}$ as in \eqref{sensing_bound}
        \STATE Compute required $n_{\com,t}$ as in \eqref{com_bound}
        \IF{$n_{\com,t}>N$ }
        \STATE $n_{\com,t}^* = N, n_{\s,t}^* = 0$
        \ELSIF{$n_{\com,t}<N$ and $n_{\com,t}<N + n_{\s,t}>N$}
        \STATE $n_{\com,t}^* = n_{\com,t}, n_{\s,t}^* = N- n_{\com,t}^*$
        \ELSE 
        \STATE $n_{\com,t}^* = n_{\com,t}, n_{\s,t}^* = n_{\s,t}$
        \ENDIF\\
        % \vspace{3pt}
         % \textit{ISAC update:}
        \STATE Update $\hat{\bs}_t$
        \STATE Update throughput $R_{\com,t}$
    }
\end{algorithmic}
\end{algorithm}
\section{Subcarier Optimization For DT}
% \vspace{-10pt}
The subcarrier allocation for communication and sensing will be based on: $(i)$ the acceptable level of accuracy ${\eta}_t$ for the estimated state, as determined by the RL agent; $(ii)$ the accuracy requirements of the DT model as specified in \eqref{qos_condition}; and $(iii)$ the availability of wireless resources, determined by the system capacity and data rate requirements (2).

% \vspace{-10pt}
\subsubsection{Reformulation Problem }
We introduce arbitrary variables $\bar{\xi}_{t}^2$  indicates desired DT's error level of system's state at QI $t$.  Given the reliability for the DT in \eqref{qos_condition} and the requisite level of accuracy ${\eta}_t$ to uphold the precision of the RL model, the DT should meet the error constraints at QI $t$ as $\sigma_{x,t}^2 \leq \bar{\xi}_{t}^2 \triangleq \min\Big\{\xi^2, \frac{1}{\eta_{t}}\Big\}$. 
% \begin{equation}\label{glob_qos}
%     \sigma_{x,t}^2 \leq \bar{\xi}_{t}^2 \triangleq \min\Big\{\xi^2, \frac{1}{\eta_{t}}\Big\}, \forall t\in\Tcal.
% \end{equation}
At the QI $(t-1)$ we solve the optimization problem
\begin{subequations} \label{reformulate_problem}
    \begin{alignat}{2}
	\text{P1}: \underset{n_{\s, t}, \ n_{\com, t}}  { \mathrm{minimize}} \ & \alpha_1\max\left\{\sigma_{x,t}^2-\bar{\xi}^2,0\right\} +\alpha_2(n_{\s,t} + n_{\com,t})\label{} \nonumber\\
    & + \alpha_3 |R_t(n_{s,t})-\bar{R},0|\label{reformulate_problema}\\
	\mathrm{s.t.} \quad  &  n_{\s, t} + n_{\com,t} \leq N, \label{reformulate_problemb}
    \end{alignat}
\end{subequations}
wherein the non-negative parameter $\alpha_1, \alpha_2, \alpha_3\in[0,1]$ represents the relative weight to accuracy and energy efficiency within the underlying objective function. The objective \eqref{reformulate_problema} relaxes the constraints \eqref{glob_problemc} and \eqref{glob_problemd} due to its dependence on practical conditions, i.e., in situations where the error surpasses a certain threshold, even ultilizing all subcarrier capacity fails to guarantee the desired reliability $\bar{\xi}^2$.

\subsubsection{Communication subcarrier allocation}
The required $n_{\com,t}^*$ to satisfied the communication constraint $R_t(n_{\com,t})\geq\bar{R}$ is obtained in closed-form by definition as
\begin{IEEEeqnarray}{lll}\label{com_bound}
    n_{\com,t}^* \geq \frac{\bar{R}}{\bar{\tau}_t\Delta f\log_2(1+\gamma_{\com,t})}, \forall t\in \Tcal.
\end{IEEEeqnarray}
\begin{table}[!t]
\vspace{-6pt}
\caption{Simulation Parameters}
\resizebox{9cm}{!} 
{
    \begin{tabular}{ll|ll}
        \hline
        Parameter & Value & Parameter & Value \\
        \hline
        Carrier frequency ($f_c$) & 28 GHz & Subcarrier spacing ($\Delta f$) & 120 kHz \\
        Bandwidth ($B$) & 1600 MHz & Symbol duration ($T_0$) & 8.92 $\mu$s \\
        Number of subcarriers ($N$) & 512 & Number of symbols ($M$) & 128 \\
        Noise figure ($F$) & 8 dB & Transmit antenna gain ($G_T$) & 33 dB \\
        Receive antenna gain ($G_E$) & 3 dBi & Transmit power ($P_{T}$) & 25 dBm \\
        The reward weight ($\kappa$) & $5\times 10^{-6}$ & Antenna spacing ($\Delta r, \Delta c$) & 0.5$\lambda$, 0.5$\lambda$ \\
        Data rate threshold $\bar{R}$ & 1000 [Mbps] & ($N_\mathrm{Per}, M_\mathrm{Per}$) & ($6400, 5120$)\\
        \hline\vspace{-1cm}
    \end{tabular}
}
\label{table_parameters}
\end{table}
\subsubsection{Sensing subcarrier allocation}
We emphasize that the target position uncertainty can be calculated in closed form based on the range \eqref{eq:AccuracyRange} and angle \eqref{eq:AccuracyElevation} uncertainties, as established by the following lemma.
\begin{lemma}\label{lemma:Position}
From the distributions of range and elevation angle, specified as $r_t$ with mean $\bar{r}_t$ and variance $\sigma_{r,t}^2$, and $\theta_t \sim \mathcal{N}(\bar{\theta}_t, \sigma_{\theta,t}^2)$ respectively, the position of the AGV $x_t$, with its mean $\bar{x}_t$ and variance $\sigma_{x,t}^2$  follows
\begin{align} \label{eq:Position}
    \left\{\begin{aligned}
        \bar{x}_t &= \bar{r}_t \cos(\bar{\theta}_t) e^{-\sigma_{\theta,t}^2/2},\\
        \sigma_{x,t}^2 &= \Gamma_t + \sigma_{r,t}^2 \Upsilon_t,
    \end{aligned}\right.
\end{align}
where $\Gamma_t = \bar{r}_t^2\left(\frac{1}{2} + \frac{1}{2}\cos(2\bar{\theta}_t)e^{-2\sigma_{\theta,t}^2} - \left(\cos(\bar{\theta}_t)e^{-{\sigma}_{\theta,t}^2/2}\right)^2 \right)$, $\Upsilon_t = \frac{1}{2} + \frac{1}{2}\cos(2\bar{\theta}_t)e^{-2\sigma_{\theta,t}^2}$. 
% \begin{align}
%     \Gamma_t &= \bar{r}_t^2\left(\frac{1}{2} + \frac{1}{2}\cos(2\bar{\theta}_t)e^{-2\sigma_{\theta,t}^2} - \left(\cos(\bar{\theta}_t)e^{-{\sigma}_{\theta,t}^2/2}\right)^2 \right), \label{eq:Gamma} \\
%     \Upsilon_t &= \frac{1}{2} + \frac{1}{2}\cos(2\bar{\theta}_t)e^{-2\sigma_{\theta,t}^2}. \label{eq:Psi}
% \end{align}
\end{lemma}
\begin{proof}
    Please see the Appendix A.
\end{proof}
The mean value and variance of the AGV position are obtained from Lemma~\ref{lemma:Position}. With $\sigma_{x,t}^2\leq\delta^2$, the sensing constraint can be formulated by $\Gamma_t + \sigma_{r,t}^2 \Upsilon_t \leq \bar{\xi}^2$, 
% \begin{IEEEeqnarray}{lrll}
%     &\Gamma_t + \sigma_{r,t}^2 \Upsilon_t \ &\leq \bar{\xi}^2 \\
%     \Leftrightarrow& \Gamma_t + \frac{6}{(n_{\s,t}^2-1)\gamma_{\s,t}} \frac{c^2}{(4\pi\Delta f)^2} \ &\leq \bar{\xi}^2,
% \end{IEEEeqnarray}
leading to the closed-form lower bound of $n_{\s,t}$ as
\begin{IEEEeqnarray}{lll}\label{sensing_bound}
    n_{\s,t} \ &\geq  \sqrt{\frac{6c^2\Upsilon_t}{(\bar{\xi}^2 -\Gamma_t)(4\pi\Delta f)^2\gamma_{\s,t}} + 1},
\end{IEEEeqnarray}
which can solveable by selecting the optimal $n_{\s,t} $ that satisfing both \eqref{sensing_bound} and the subcarrier capacity $N$. The proposed algorithm for solving \eqref{glob_problem} is outlined in Alg.~1. 
\begin{figure}[t]
    \centering
    \includegraphics[width = 0.5\textwidth]{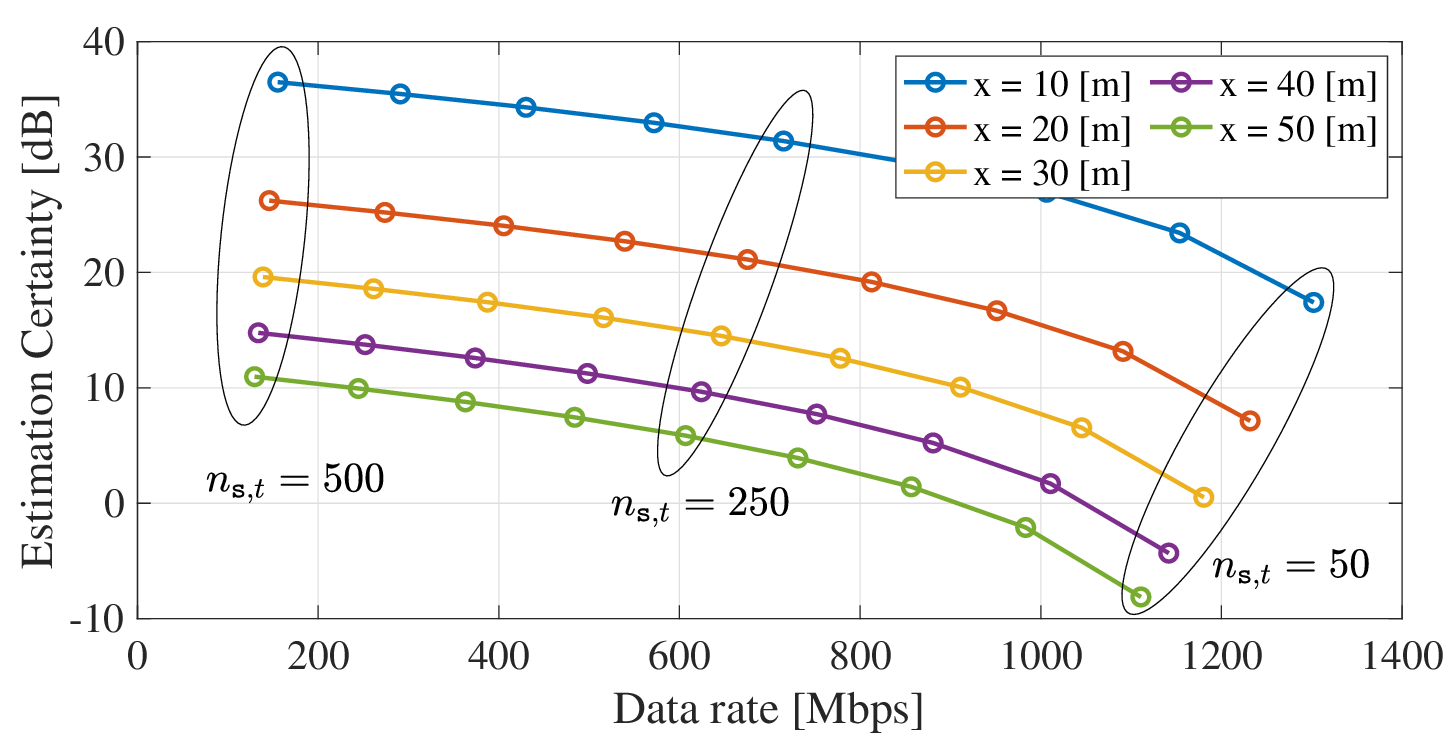}
    \captionsetup{format=plain, justification=justified, width=1\linewidth}
    \caption{\small{Trade-off analysis between sensing certainty and communication data rate at various distances.}}
    \label{fig:trade_off}
    \vspace{-20pt}
\end{figure}
% \begin{figure}[t]
%     \centering
%     \includegraphics[width = 0.5\textwidth]{figs/fig2_control_v2.eps}
%     \captionsetup{format=plain, justification=justified, width=1\linewidth}
%     \caption{\small{The snapshot of uncertainty evolution and no. offered sensing subcarriers in Proposed-SP scheme vs QI to meet both the uncertainty requirements of DT and control solution.}}
%     \label{fig:dynamic_evolution}
%     \vspace{-15pt}
% \end{figure}

% The proposed algorithm for solving \eqref{glob_problem} được trình bày steps trong algorithm 1. Cụ thể, tại mỗi QI, chúng tôi first perform optimal action dựa trên estimated state $\bs_{t-1}$. Giả định này hợp lý trong bối cảnh chúng tôi xem xét AGV moving within inside manufactory, với tốc độ tối đa thông thường 1.7m/s và each QI được set up vào 1ms. Điều này dẫn tới việc vị trí của AGV không thay đổi nhiều. 
% The proposed algorithm for solving \eqref{glob_problem} is outlined in Alg.~1. Specifically, at each QI, we first perform the optimal action based on the estimated state $\hat{\bs}_{t}$. This assumption is reasonable within the context of an AGV moving inside a manufacturing facility, typically at a maximum speed of 1.7 m/s, with each QI as 1 ms. Consequently, the position of the AGV does not change significantly within each interval.
%%%%%%%%%%%%%%%%%%%%%%%%%%%%%%%%%%%%%%%%%%%%%%%%
\section{Simulation results}
%%%%%%%%%%%%%%%%%%%%%%%%%%%%%%%%%%%%%%%%%%%%%%%%

In this section, we provide the numerical results to investigate our findings and proposed design. The important parameters are listed in Table \ref{table_parameters}. Unless otherwise noted, at each QI, the AGV is served by the nearest AP with a distance randomly ranging from 5 to 30 meters, and the DT requirement $\delta_t = 0.02\ \forall t$. In Fig.~\ref{fig:trade_off}, we plot the trade-off between sensing certainty, $1/\sigma^2_{x}$, and achievable data rate, $R_{\com,t}$, for different positions of the AGV. The number of subcarrier for communication with each value of $n_{\s,t}$ is $n_{\com,t} = N - n_{\s,t}$. It is evident that as more communication resources are allocated to sensing, the accuracy of the estimated position improves, while the AGV's throughput correspondingly decreases. Specifically, for a typical distance of $20 m$, if the AP allocates 250 subcarriers to sensing, the position estimation certainty is 7.5 dB, and the corresponding data rate is $600$ Mbps. Decreasing $n_{\s,t}$ to 50 results in an AGV throughput of approximately 1100 Mbps, while the position uncertainty of the estimate is -9 dB. Recognizing this trade-off is crucial for optimizing the system in various applications with differing sensing and communication requirements.
\begin{figure}[t]
    \centering
    \includegraphics[width = 0.5\textwidth]{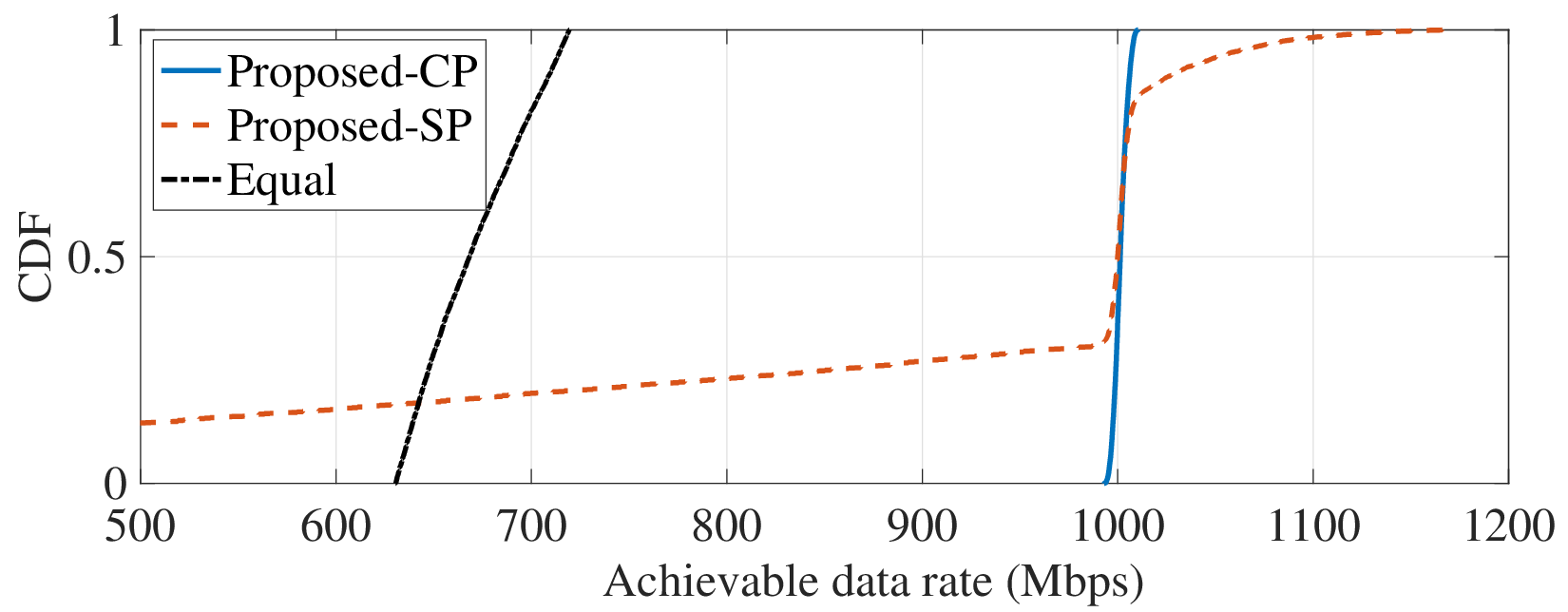}
    \caption{\small{The achievable data rate of different schemes.}}
    \label{fig:CDF_rate}
    \vspace{-15pt}
\end{figure}

% Fig.~\ref{fig:dynamic_evolution} presents a snapshot of the uncertainty evolution and the strategic allocation of subcarriers for sensing, based on the combined certainty $\bar{\delta}_t^2$ , which integrates  requirements from DT $\xi^2_t$ and RL agent $\eta_t$. The constraints on the RL agent's behavior are primarily governed by the reward function $\tilde{r}_t$. As a result, the DT demands higher position estimation accuracy only when the DT threshold is exceeded or when the AGV are on critical positions, necessitating precise force control. This strategy  conserves resources while enhances the system's ability to accurately capture and compute control values.

The CDF of the achievable data rate among different algorithms is presented in Fig.~\ref{fig:CDF_rate} and their corresponding number of QI to reach the target in Fig.~\ref{fig:CDF_QI}. In these figure,  \textit{Proposed-CP} (Proposed Communication Priority), depicted in Alg.~1,  prioritizes allocating $n_{\com,t}$ to meet $\bar{R}$ before allocating for $n_{\s,t}$. \textit{Proposed-SP} (Proposed Sensing Priority), also follows Alg.~1 but prioritizing subcarriers for sensing the AGV's position before allocating them for communication. The \textit{Equal} method allocates subcarriers equally for sensing and communication as $N/2$ at all times. Compared to the baseline, our schemes demonstrate superior performance in serving the AGV communication needs.  It is noted that although the Proposed-CP achieves better communication performance, with $R_{\com,t}$ meeting $\bar{R}$ most of the time compared to \textit{Proposed-SP}, which fails to meet $\bar{R}$ nearly $30\%$ of the time, the AGV takes longer to reach the target due to sensing requirements not meeting the RL request as in Fig.~\ref{fig:CDF_QI}.

% \vspace*{-10pt}
%%%%%%%%%%%%%%%%%%%%%%%%%%%%%%%%%%%%%%%%%%%%%%%%
\section{Conclusions}
%%%%%%%%%%%%%%%%%%%%%%%%%%%%%%%%%%%%%%%%%%%%%%%%
% \vspace*{-10pt}
This paper analyzes the impact of ISAC design within a DT framework for optimal system monitoring and control. By calculating the number of subcarriers required for sensing object positions and meeting communication requirements, we can easily propose options for efficient and economical wireless resource allocation. Our future work aims to extend the DT using ISAC in complex and practical environments with multiple AGVs moving and computing in indoor settings. %Additionally, we will address the context of multiple APs causing mutual interference, requiring resource allocation to meet positioning and communication demands.
\vspace*{-15pt}
\section*{Appendix A}
Lets $x_t$ follows a distribution with mean value as $\bar{x}_t$ and variance  $\sigma_{x,t}^2)$.
To find $\bar{x}_t$, we do have 
\begin{IEEEeqnarray}{ll}
    \bar{x}_t = \mathbb{E}[r_t\cos(\theta_t)] = \mathbb{E}[r_t]\ \mathbb{E}[\cos(\theta_t)] = R_0 \cos(\bar{\theta}_t) e^{-\sigma_{\theta,t}^2/2}. \nonumber
\end{IEEEeqnarray}
Next, to compute the variance $\sigma_{x,t}^2$, we have the fact that
\begin{IEEEeqnarray}{lll}\label{eq:var_x}
    \sigma_{x,t}^2 &= \mathbb{E}[r_t]^2\mbox{Var}[\cos(\theta_t)] + \mbox{Var}[r_t]\mathbb{E}[\cos^2(\theta_t)] \nonumber \\
    & = \bar{r}_t^2 \mbox{Var}[\cos(\theta_t)]  + \sigma_{r,t}^2\mathbb{E}[\cos^2(\theta_t)].
\end{IEEEeqnarray}
In \eqref{eq:var_x}, the goal is to convert $\mathbb{E}[\cos^2(\theta_t)]$ and $\mbox{Var}[\cos(\theta_t)]$ to tractable forms.  First, $\E[\cos^2(\theta_t]$ can be equivalently reformulated as $\E[\cos^2(\theta_t]  = \E\left[0.5({1+\cos(2\theta_t)})\right]
     = 0.5 + 0.5\E[\cos(2\theta_t)]
     = 0.5+ 0.5\cos(2\bar{\theta}_t)e^{-2\sigma^2_{\theta,t}}$. 
Using this result, we can rewrite $\mbox{Var}[\cos(\theta_t)]$ to tractable manner as $\mbox{Var}[\cos(\theta)] = \E[\cos^2(\theta)] - (\E[\cos(\theta)])^2  = \frac{1}{2} + \frac{1}{2}\cos(2\theta_0)e^{-2\sigma^2_\theta} - \left(\cos(\theta_0)e^{-\sigma_\theta^2/2}\right)^2.$
Summing up, the $\sigma_{x,t}^2$ is given by
\begin{IEEEeqnarray}{lll}
    \sigma_{x,t}^2 =& \bar{r}_t^2\left(\frac{1}{2} + \frac{1}{2}\cos(2\bar{\theta}_t)e^{-2\sigma^2_{\theta,t}} - \left(\cos(\bar{\theta}_t)e^{-\sigma_{\theta,t}^2/2}\right)^2 \right) \nonumber\\
    &+ \sigma_{r,t}^2 \left(\frac{1}{2} + \frac{1}{2}\cos(2\bar{\theta}_t)e^{-2\sigma^2_{\theta,t}}\right).
\end{IEEEeqnarray}
By denoting $\Gamma_t$ and $\Upsilon_t$ as in \eqref{eq:Position}, we obtain $\sigma_{x,t}^2$ as in \eqref{eq:Position} and complete the proof.

\begin{figure}[t]
    \centering
    \includegraphics[width = 0.5\textwidth]{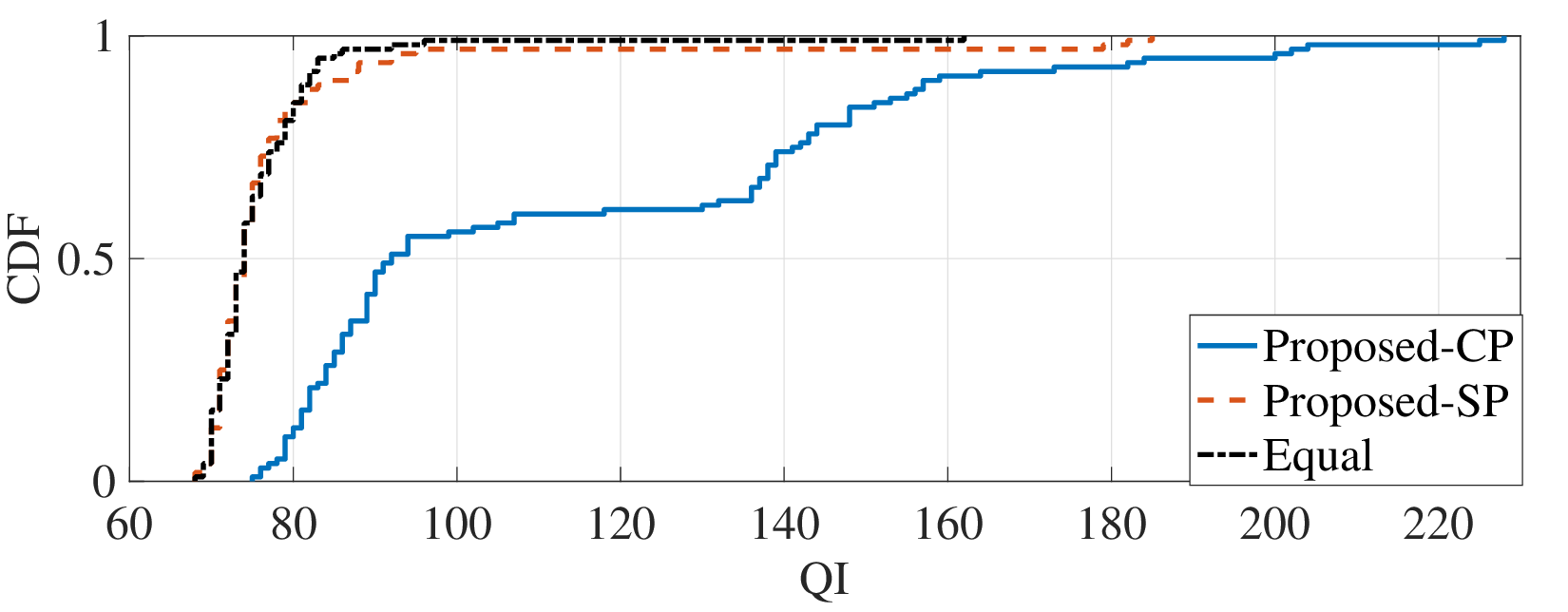}
    \caption{\small{Number of QI for AGV to reach the target.}}
    \label{fig:CDF_QI}
    \vspace{-15pt}
\end{figure}

\setstretch{0.95}
% \vspace*{-10pt}
\bibliographystyle{IEEEtran}
\vspace*{-10pt}
% \balance
\bibliography{Journal}
%\vspace{-0.2cm}
\end{document}